\documentclass[pra,floatfix,amsmath,twocolumn]{revtex4}
\usepackage{amssymb}
\usepackage{graphicx}
\usepackage{graphics}
\usepackage{amsmath}
\usepackage{amsthm}
\usepackage{color}

\def\x{\mathcal{X}}
\def\y{\mathcal{Y}}

\def\b{\mathcal{B}}

\newcommand{\bea}{\begin{eqnarray}}
\newcommand{\eea}{\end{eqnarray}}

\def\bi{\begin{itemize}}
\def\ei{\end{itemize}}
\def\bc{\begin{center}}
\def\ec{\end{center}}

\def\C{\hbox{$\mit I$\kern-.7em$\mit C$}}
\def\R{\hbox{$\mit I$\kern-.6em$\mit R$}}

\newcommand{\one}{\mbox{$1 \hspace{-1.0mm}  {\bf l}$}}
\def\tr{\mathrm{tr}}

\def\c{\mathcal{C}}
\def\q{\mathcal{Q}}

\def\P{\textbf{P}}

\def\supp{\textrm{supp\,}}

\newtheorem{theorem}{Theorem}

\newtheorem{lemma}[theorem]{Lemma}

\begin{document}

\title{A general bound for the dimension of quantum behaviours in the prepare-and-measure scenario}
\author{Julio I. de Vicente} \email{jdvicent@math.uc3m.es}
\affiliation{Departamento de Matem\'aticas, Universidad Carlos III de
Madrid, Avda. de la Universidad 30, E-28911, Legan\'es (Madrid), Spain}

\begin{abstract}
The prepare-and-measure scenario offers the possibility to infer the dimension of an unknown physical system in a device-independent way, i.e.\ using only raw measurement data with apparatuses regarded as black boxes. We provide here a general lower bound on the dimension necessary to observe arbitrary quantum behaviours in this scenario based on simple matrix analysis. This bound holds even if the preparer and measurer share randomness. This is relevant in scenarios were the parties are free to access this resource or it is not safe to assume that the devices are not correlated. We further use this result to bound the success probability of random access codes in general as a function of the dimension of the quantum systems sent from one party to another and we provide constructions of dimension witnesses.
\end{abstract}

\maketitle

\section{Introduction}

The device independent approach in quantum information theory allows one to infer physical properties of systems and to implement protocols based solely on the observed statistics, i.e.\ without making any assumption on the underlying states and how they interact with the measurement apparatuses. Different tasks that can be implemented in this way include quantum key distribution, randomness generation and amplification, genuine multipartite entanglement certification and self-testing of states and measurements \cite{reviewnl}. It has been also observed \cite{bellwit,gallego} that the underlying dimension of an uncharacterized (classical or quantum) physical system can be tested in this way, i.e.\ using only the observed probabilities of obtaining certain outcomes conditioned on implementing different uncharacterized measurements. These procedures are referred to as device-independent dimension witnessing (DIDW) and are the object of active current investigation. On the one hand, from a foundational perspective, DIDW allows to estimate the degrees of freedom of a system without a priori including this information in the physical model. On the other hand, from the point of view of applications, quantum information tasks can be more efficiently implemented the larger the dimensions of the quantum systems one can prepare and control. Thus, dimension is regarded as a valuable resource in this context and DIDW provides experimental means to test it. In fact, DIDW can be regarded as a primitive for semi-device-independent protocols, which make no assumption on the inner functioning of devices and physical systems except for bounds on the underlying dimension \cite{semi}. In general, DIDW is deeply rooted in the field of quantum communication complexity, which studies the necessary amount of communication different parties must exchange (as measured by the dimensionality of the physical systems being sent) in order to implement distributed computations \cite{reviewcc}.

DIDW was introduced in the Bell scenario in which two parties share an entangled state \cite{bellwit}. Soon after, Ref.\ \cite{gallego} presented an alternative scenario, the so-called prepare-and-measure, which consists of two devices: one that prepares and sends states and one that measures them. This setting is simpler in the sense that it does not require entanglement nor multicomponent systems and this proposal has been already verified in experiments \cite{experiments}. Reference \cite{gallego} analyzed the mathematical structure of the set of possible behaviours to be observed in this scenario depending on the dimension and provided explicit constructions of functionals acting there, known as dimension witnesses, whose values provide lower bounds on the classical and quantum dimension. Although other constructions of dimension witnesses have appeared in subsequent works \cite{dallarno,didwpm,didwseveral}, this approach suffers from two difficulties. First, these functionals are usually tailor-made to detect specific behaviours one targets at. Second, the explicit corresponding bounds are very case-dependent and difficult to find in general. Thus, general bounds to constrain the dimension of arbitrary quantum behaviours are of great use in this context. In this sense, it is particularly worth mentioning the work of Ref.\ \cite{num}, which offers a powerful and general numerical approach based on semidefinite programming that makes it possible to obtain such bounds for a given dimension witness. However, this approach is bound to problems that can be tackled numerically. Analytical results in this direction not only enable a better understanding of the mathematical structure of dimension-constrained behaviours but also to obtain results for problems that go beyond computational efficiency such as the asymptotic scenario, to devise new dimension witnesses or to consider nonlinear constraints. In this sense, general lower bounds on the dimension can be found in \cite{brunner,sikora,yo}. Notwithstanding, all these analytical approaches assume that the preparer and measurer devices are not correlated, i.e.\ that they do not share a random variable. Although this is justified in certain scenarios, there are others in which this assumption is not admissible. This is particularly the case when the devices are not trusted. Suppose, for instance, that the parties want to verify that devices provided by a manufacturer, which are regarded by them as black boxes, operate on quantum systems of a given dimension. Then, malicious providers could fake higher-dimensional behaviours by mixing lower-dimensional preparations using shared randomness if this is not taken into account. Another example to consider this scenario is when the parties have to use their devices to implement a particular task and the constraints of the problem allow them to use shared randomness as a resource. In fact, it is known that the availability or not of shared randomness can have drastic consequences in what comes to the necessary dimension underlying a given observed behaviour; for instance, without this resource at disposal almost all behaviours are high-dimensional while, when it is given, low-dimension behaviours are no longer negligible \cite{yo}. The main goal of the present work is to close this gap and provide general analytical conditions to test the dimension of quantum behaviours in the prepare-and-measure scenario even when devices might share randomness. After presenting our notation and definitions in Sec.\ II, in Sec.\ III we provide a general lower bound on the dimension of a quantum behaviour based on simple matrix analysis techniques, which is valid when preparations can be mixed either because shared randomness is available as a resource or because it is not safe to assume that the devices are uncorrelated. To further illustrate its usefulness, we provide two applications of this result in Sec.\ IV: we obtain bounds on the efficiency of a communication-complexity protocol known as random access coding and we provide improved constructions of dimension witnesses. We finish in Sec.\ V with some concluding remarks.

\section{Notation and definitions}

The prepare-and-measure scenario for DIDW \cite{gallego,dallarno} is composed by two parties: the preparer, Alice (or A), and the measurer, Bob (or B). A and B receive respectively inputs $x$ and $y$ from finite alphabets $\x$ and $\y$. They can only communicate by A sending a classical or quantum physical system to B depending on her input $x$. B can then perform a measurement on the system he receives depending on his input $y$. Using the outcome of this measurement together with all previous information held by him, B produces then an output $b$, which also takes values from a finite alphabet $\b$. Moreover, as explained in the introduction, we assume that the devices held by A and B may be correlated. That is, both parties have access to a common random variable, whose value determines the strategy to be followed from a pre-established list available to them. The main object in this scenario is the conditional probabilities with which each output occurs for any given pair of inputs: $P(b|xy)$. We will refer to this object as behaviour and it will be denoted by $\P$. Behaviours are lists of real numbers characterized by $P(b|xy)\geq0$ $\forall b,x,y$ and $\sum_bP(b|xy)=1$ $\forall x,y$ due to the fact that they are a collection of conditional probability distributions.

Suppose now that an observer can monitor sufficient repetitions of this process so as to infer the corresponding behaviour but has no information about the systems sent by Alice and the measurements implemented by Bob nor about any details of the strategy the parties use to determine the output corresponding to the different possible inputs. The task DIDW aims at is to determine the minimal amount of classical or quantum communication (as quantified by the dimension of the systems sent from A to B) that is compatible with the observed behaviour. In order to give a rigorous definition of this quantity we distinguish between the cases in which A sends classical or quantum states. In the first case, A will send a message $m(x)\in[d]=\{1,\ldots,d\}$, and the number of dits $d$ necessary to construct it quantifies the amount of classical communication. The availability of shared randomness boils down to the fact that the parties can prepare any convex combination of strategies using messages of dimension less than or equal to $d$. We denote the set of all such behaviours by $\c_d$ (this set and the analogous for the quantum case to be defined below depend on $|\x|$, $|\y|$ and $|\b|$, but we do not make this explicit in order to ease the notation as these quantities should be in general clear from the context). Notice that the availability of shared randomness imposes that the set $\c_d$ is convex. Furthermore, it can be seen that this set is actually a convex polytope \cite{gallego,dallarno}: it is the convex hull of a finite number of behaviours $\{P^D_i\}$, which we call deterministic. These have the structure
\begin{equation}\label{setc}
P^D(b|xy)=\sum_{m=1}^{d}s(m|x)t(b|my),
\end{equation}
where $s(m|x)$ codifies the conditional probability with which A sends the message $m$ given $x$, and $t(b|ym)$ the conditional probability with which B outputs $b$ given $y$ and the reception of $m$. The deterministic condition amounts to the fact that $s(m|x)=\delta_{m,f(x)}$ and $t(b|my)=\delta_{b,g(m,y)}$ with arbitrary functions $f:\x\to[d]$ and $g:[d]\times\y\to\b$. Considering all possible choices for these functions gives rise to the finite list $\{P^D_i\}$.

In the quantum case A sends quantum states $\rho_x$. The dimension of her message is thus
\begin{equation}\label{quantumdim}
d=\dim \sum_x\supp\rho_x,
\end{equation}
where $\supp$ stands for the support of an operator. In order to produce his output, B interacts with the state he receives by choosing an arbitrary quantum measurement conditioned on his input. Thus, the set $\q_d$ of behaviours achievable by sending quantum states of dimension at most $d$ is given by the convex hull of all behaviours $\P$ that take this form: there exists measurements $\{\Pi_b^y\geq0\}$ with $\sum_b\Pi_b^y=\one$ $\forall y$, such that
\begin{equation}\label{setq}
P(b|xy)=\tr(\rho_x\Pi_b^y)
\end{equation}
where the $\{\rho_x\}$ are of dimension less than or equal to $d$ as given by Eq.\ (\ref{quantumdim}).

One can readily find that $\c_{|\x|}=\q_{|\x|}$, which constitute the set of all behaviours in a given setting. This is because if $d=|\x|$, A can transmit to B the value of her input through her message. Therefore, given any observed behaviour $\P$ there always exist minimal values of $1\leq d\leq|\x|$ and $1\leq d'\leq|\x|$ such that $\P\in\c_d$ and $\P\in\q_{d'}$ (notice that in general $d'\leq d$ since it is straightforward to see that for any fixed value of $d$, it holds that $\c_d\subseteq\q_d$). To determine them is precisely the goal of DIDW. The fact that the sets $\c_d$ are polytopes provides techniques to bound the classical dimension necessary to observe a given behaviour \cite{gallego}. However, the quantum case is much harder to deal with. In the following we provide such a bound in terms of a simple function of the behaviour.

\section{Main result}

We will arrange the array of numbers given by $\P$ into a matrix $P\in\mathbb{R}^{|\x|\times|\y||\b|}$ according to the rule
\begin{equation}\label{behavior}
P=\sum_{bxy}P(b|xy)|x\rangle\langle yb|,
\end{equation}
where in the standard notation of quantum mechanics $|yb\rangle=|y\rangle\otimes|b\rangle$ and $\{|y\rangle\}$ denotes the computational basis of $\mathbb{R}^{|\y|}$ and similarly for the other alphabet elements. %In other words, $P$ takes the form
%\begin{widetext}
%\begin{equation}
%P=\left(
%    \begin{array}{ccccccc}
%      P(0|11) & P(1|11) & P(0|12) & P(1|12) & \cdots & P(0|1|\y|) & P(1|1|\y|) \\
%      P(0|21) & P(1|21) & P(0|22) & P(1|22) & \cdots & P(0|2|\y|) & P(1|2|\y|) \\
%      \vdots & \vdots & \vdots & \vdots & \ddots & \vdots & \vdots \\
%      P(0||\x|1) & P(1||\x|1) & P(0||\x|2) & P(1||\x|2) & \cdots & P(0||\x||\y|) & P(1||\x||\y|) \\
%    \end{array}
%  \right).
%\end{equation}
%\end{widetext}
We will consider different Schatten norms for matrices:
\begin{equation}\label{schatten}
||A||_p=\left(\sum_i\sigma_i^p(A)\right)^{1/p}\quad(1\leq p\leq\infty),
\end{equation}
where $\{\sigma_i(A)\}$ are the singular values of the matrix $A$. Finally, we will denote the standard Hilbert-Schmidt inner product of matrices by
\begin{equation}
\langle P,G\rangle=\tr(PG^T)=\sum_{bxy}P(b|xy)G(b|xy),
\end{equation}
where for an arbitrary collection of $|\x||\y||\b|$ real numbers $G(b|xy)$, we define the matrix $G$ following the same prescription as in Eq.\ (\ref{behavior}).

\begin{theorem}
In any prepare-and-measure scenario $(|\x|,|\y|,|\b|)$, if $\P\in\q_d$ then
\begin{equation}
d\geq\frac{||P||_1^2}{|\x||\y|}.
\end{equation}
\end{theorem}
\begin{proof}
Due to the triangle inequality, the maximal value of $||P||_1$ in $\q_d$ must correspond to behaviours of the form given by Eq.\ (\ref{setq}) and $\rho_x\in\mathbb{C}^{d\times d}$ $\forall x$, i.e.\ shared randomness can be ignored. Defining the matrix
\begin{equation}
Z=\sum_{bxy}|x\rangle\langle yb|\otimes\rho_x\Pi_b^y\in\mathbb{R}^{|\x|\times|\y||\b|}\otimes\mathbb{C}^{d\times d},
\end{equation}
we have that $P=\tr_{\mathbb{C}^{d\times d}} Z$. Thus, since the trace norm cannot increase by partial tracing, one arrives at
\begin{equation}
||P||_1\leq||Z||_1\leq\left|\left|\sum_x|x\rangle\otimes\rho_x\right|\right|_2\left|\left|\sum_{by}\langle yb|\otimes\Pi_b^y\right|\right|_2,
\end{equation}
where in the last step we have used a particular case of H\"{o}lder's inequality for Schatten norms (see e.g.\ \cite{bhatia}). The result follows by noticing that
\begin{equation}
\left|\left|\sum_x|x\rangle\otimes\rho_x\right|\right|_2\leq\sum_x\left|\left||x\rangle\otimes\rho_x\right|\right|_2=\sum_x\tr(\rho_x^2)\leq|\x|,
\end{equation}
where we have used that $\tr(\rho_x^2)\leq1$ $\forall x$, and
\begin{align}
\left|\left|\sum_{by}\langle yb|\otimes\Pi_b^y\right|\right|_2&\leq\sum_{by}\left|\left|\langle yb|\otimes\Pi_b^y\right|\right|_2=\sum_{by}\tr[(\Pi_b^y)^2]\nonumber\\
&\leq\sum_{by}\tr(\Pi_b^y)=d|\y|,
\end{align}
%\end{equation}
where we have used that $0\leq\Pi_b^y\leq\one$ $\forall b,y$.
\end{proof}

\begin{theorem}
In any prepare-and-measure scenario $(|\x|,|\y|,|\b|)$ and for every matrix $G\in\mathbb{R}^{|\x|\times|\y||\b|}$, if $\P\in\q_d$ then
\begin{equation}
\langle P,G\rangle\leq||G||_\infty\sqrt{d|\x||\y|}.
\end{equation}
\end{theorem}
\begin{proof}
This follows readily from Theorem 1 by another particular case of H\"older's inequality for Schatten norms ($\langle A,B\rangle\leq||A||_\infty||B||_1$).
\end{proof}

It should be noticed that Theorems 1 and 2 are equivalent since the former can also be deduced from the latter by noticing that $||P||_1=\max_U\tr(PU)$ where the maximization is over all partial isometries in $\mathbb{R}^{|\y||\b|\times|\x|}$ (and, hence, $||U||_\infty=||U^T||_\infty=1$) \cite{HJ2}. Theorem 1 provides a directly checkable condition that allows to lower bound the quantum dimension necessary to observe any given behaviour without the need of a clever choice for the matrix $G$. However, any linear functional acting on the set of behaviours takes the form $\langle P,G\rangle$ for some $G$ and, thus, Theorem 2 is also of interest since, among other applications that we shall discuss in the next section, it provides upper bounds within $\q_d$ for the so-called dimension witnesses. %It is nevertheless remarkable that, as we shall see when we discuss the aforementioned applications, it seems that no or very little price is paid to turn Lemma 2 into the easy-to-use condition of Theorem 1: Lemma 2 appears to be most stringent precisely when $G$ is a partial isometry and given a extremal behaviour, the singular value decomposition (SVD) of $P$ seems to lead to a powerful witness (notice that if $P$ has SVD $P=U\Sigma V^T$, then $\langle P,G\rangle=||P||_1$ if $G=UV^T$).

Before considering these applications, let us first discuss the attainability of the bound given in Theorem 1. It turns out that the bound cannot be improved in general since there exist scenarios in which it is sharp $\forall d$. In particular, it suffices to consider deterministic behaviours in $\c_d$. For this, let $|\x|=dn$, $|\y|=m$ and $|\b|=d$ for any $d,m,n\in\mathbb{N}$ and let us introduce the notation $\textbf{1}_n$ and $\textbf{0}_n$ for the vectors in $\mathbb{R}^n$ that have all entries equal to 1 and 0 respectively and $e_i^{(n)}$ for the vector of $\mathbb{R}^n$ that has zeroes everywhere except a 1 in the $i$th entry. Take then the behaviour $\P$ with matrix
%\begin{widetext}
\begin{align}
P&=\left(
    \begin{array}{c}
      \textbf{1}_n \\
      \textbf{0}_n \\
      \textbf{0}_n \\
      \vdots \\
      \textbf{0}_n \\
    \end{array}
  \right)\left(
           \begin{array}{cccc}
             (e_1^{(d)})^T & (e_1^{(d)})^T & \cdots & (e_1^{(d)})^T \\
           \end{array}
         \right)\nonumber\\&+\left(
    \begin{array}{c}
      \textbf{0}_n \\
      \textbf{1}_n \\
      \textbf{0}_n \\
      \vdots \\
      \textbf{0}_n \\
    \end{array}
  \right)\left(
           \begin{array}{cccc}
             (e_2^{(d)})^T & (e_2^{(d)})^T & \cdots & (e_2^{(d)})^T \\
           \end{array}
         \right)\nonumber\\&+\cdots+\left(
    \begin{array}{c}
      \textbf{0}_n \\
      \textbf{0}_n \\
      \vdots \\
      \textbf{0}_n \\
      \textbf{1}_n \\
    \end{array}
  \right)\left(
           \begin{array}{cccc}
             (e_d^{(d)})^T & (e_d^{(d)})^T & \cdots & (e_d^{(d)})^T \\
           \end{array}
         \right),
\end{align}
%\end{widetext}
where the column (row) vectors belong to $\mathbb{R}^{dn}$ ($\mathbb{R}^{dm}$). It then follows that $||P||_1=d\sqrt{nm}=\sqrt{d|\x||\y|}$ and that $\P\in\c_d$ (and, hence, $\P\in\q_d$). To see the first claim notice that $P=\sum_{i=1}^d\sqrt{nm}|u_i\rangle\langle v_i|$ where the $\{|u_i\rangle\}$ and $\{|v_i\rangle\}$ are sets of orthonormal vectors. To see the second claim, notice that $P(b|xy)$ takes the form (\ref{setc}) with $s(m|x)=\delta_{m,\lceil x/d\rceil}$ and $t(b|my)=\delta_{bm}$.

It is worth remarking that, despite the above example, not all deterministic behaviours in $\c_d$ attain the bound, i.e.\ it can be easily checked that there exist such instances where $||P||_1<\sqrt{d|\x||\y|}$. Notice, however, that this does not imply that the estimation given by Theorem 1 is not optimal as it may happen that $d\neq||P||_1^2/(|\x||\y|)$ but $d=\lceil||P||_1^2/(|\x||\y|)\rceil$.

More interestingly, as we show in the next section, we can also prove that there exist $\P\in\q_d$ such that $\P\notin\c_d$ for which $||P||_1=\sqrt{d|\x||\y|}$. Certain quantum random access codes or the behaviours considered in \cite{didwpm} provide such examples.

\section{Applications}

\subsection{Quantum random access codes}

As mentioned in the introduction, DIDW is closely related to the field of communication complexity in the setting of one-way communication complexity. Here, one asks what the minimal dimension of the (classical or quantum) messages from A to B must be in order for B to compute a given function $f(x,y): \x\times\y\to\b$ with a certain degree of success. The figure of merit which is usually considered here is the worst-case probability
\begin{equation}
p_w=\min\{P(b|xy):f(x,y)=b\}.
\end{equation}
A particular instance of this problem that has received quite some attention in the literature is random access coding \cite{rac}. Although several particular versions of this protocol have been considered, here we take the most general form in which A receives a string $x=x_1\cdots x_n$ where $x_i\in\{1,\ldots,m\}$ and B receives an input $y\in\{1,\ldots,n\}$ with the goal that $f(x,y)=x_y$ (thus $|\x|=m^n$, $|\y|=n$ and $|\b|=m$). It has been shown in \cite{ozols} that if A and B have access to shared randomness, then for the corresponding optimal strategies $p_w$ equals the average success probability
\begin{equation}
p=\frac{1}{nm^n}\sum_{f(x,y)=b}P(b|xy).
\end{equation}
Thus, the probability of success of any $(m,n)$ quantum random access code (QRAC) with communication cost $d$ can be written as $\langle P,G(m,n)\rangle$ with $\P\in\q_d$ and Theorem 2 can be applied to upper bound $p$ as a function of $d$. Here, $G(m,n)=F(m,n)/(nm^n)\in\mathbb{R}^{m^n\times mn}$, where
\begin{equation}
F(m,n)=\sum_{b,x,y}\delta_{b,x_y}|x\rangle\langle yb|.
\end{equation}
It might also be helpful to have in mind the following inductive construction of this matrix,
\begin{equation}
F(m,n)=\left(
         \begin{array}{ccccc}
           \textbf{1}_{m^{n-1}} & \textbf{0}_{m^{n-1}} & \cdots & \textbf{0}_{m^{n-1}} & F(m,n-1) \\
           \textbf{0}_{m^{n-1}} & \textbf{1}_{m^{n-1}} & \cdots & \textbf{0}_{m^{n-1}} & F(m,n-1) \\
           \vdots & \vdots & \ddots & \vdots & \vdots \\
           \textbf{0}_{m^{n-1}} & \cdots & \textbf{0}_{m^{n-1}} & \textbf{1}_{m^{n-1}} & F(m,n-1) \\
         \end{array}
       \right)
\end{equation}
with $F(m,1)=\one_m$.

Before proceeding to establish the bound it should be stressed that the application of Theorem 2 is not completely straightforward as it allows for a certain form of optimization. This is because there exist different choices of matrix $G$ to codify the same function $f$ due to the fact that $\sum_bP(b|xy)=1$ $\forall x,y$. Indeed, denoting by $\{A_{xy}\}$ the matrices
\begin{equation}
A_{xy}=\sum_b|x\rangle\langle yb|,
\end{equation}
we have that $\langle P,A_{xy}\rangle=1$ $\forall x,y$ and for every behaviour $\P$. Thus, $\forall\P\in\q_d$ it holds that
\begin{equation}
\langle P,G\rangle\leq||G+\sum_{xy}\alpha_{xy}A_{xy}||_\infty\sqrt{d|\x||\y|}-\sum_{xy}\alpha_{xy}
\end{equation}
for any choice of real numbers $\{\alpha_{xy}\}$. Interestingly, these norms are sensitive to the constraint fulfilled by behaviours and they can lead to different bounds. In our case, it seems that best results are obtained when the matrix $G+\sum_{xy}\alpha_{xy}A_{xy}$ is chosen to be a partial isometry. In the following we use the notation $\textbf{1}(m,n)$ for the $m\times n$ matrix with all entries equal to one.

\begin{lemma}
Let
\begin{equation}\label{eqlemma3}
H=\frac{1}{\sqrt{m^{n-1}}}(F(m,n)-a_{mn}\textbf{1}(m^n,mn)),
\end{equation}
with
\begin{equation}
a_{mn}=\frac{1}{m}-\frac{1}{m\sqrt{n}}.
\end{equation}
Then, $H$ is a partial isometry (and, hence, $||H||_\infty=1$).
\end{lemma}
\begin{proof}
In order to prove the claim we show that all eigenvalues of $H^T H$ are either 1 or 0. We will denote by $h_i$ the columns of the matrix $H$, which means that the index takes values $i=(y,b)\in\y\times\b$. Consequently, we will say that $i$ and $j$ belong to the same input if $i=(y,b)$ and $j=(y,b')$ for some $y\in\y$. Notice that all entries of $H$ are either $(1-a_{mn})/\sqrt{m^{n-1}}$ or $-a_{mn}/\sqrt{m^{n-1}}$ and, therefore,
\begin{widetext}
\begin{equation}
h_i^Th_i=\frac{1}{m^{n-1}}\left[m^{n-1}(1-a_{mn})^2+(m-1)m^{n-1}a_{mn}^2\right]=1-\frac{1}{m}+\frac{1}{mn}:=a,
\end{equation}
and, if $i\neq j$,
\begin{equation}
h_i^Th_j=\frac{1}{m^{n-1}}\left[-2m^{n-1}(1-a_{mn})a_{mn}+(m-2)m^{n-1}a_{mn}^2\right]=-\frac{1}{m}+\frac{1}{mn}:=b
\end{equation}
if $i$ and $j$ belong to the same input while otherwise we have that
\begin{equation}
h_i^Th_j=\frac{1}{m^{n-1}}\left[m^{n-2}(1-a_{mn})^2+(m-1)^2m^{n-2}a_{mn}^2-2(m-1)m^{n-2}(1-a_{mn})a_{mn}\right]=\frac{1}{mn}.
\end{equation}
\end{widetext}
Thus, our $mn\times mn$ matrix is given by
\begin{equation}
H^TH=\left(
       \begin{array}{cccc}
         A & B & \cdots & B \\
         B & A & \ddots & \vdots \\
         \vdots & \ddots & \ddots & B \\
         B & \cdots & B & A \\
       \end{array}
     \right)
\end{equation}
with $m\times m$ blocks $B=\textbf{1}(m,m)/(mn)$ and
\begin{equation}
A=\left(
       \begin{array}{cccc}
         a & b & \cdots & b \\
         b & a & \ddots & \vdots \\
         \vdots & \ddots & \ddots & b \\
         b & \cdots & b & a \\
       \end{array}
     \right).
\end{equation}
Since $H^TH$ happens to be a row stochastic matrix, it follows that $\textbf{1}_{mn}$ is an eigenvector with eigenvalue 1. On the other hand, $A$ is a circulant matrix, so it is easily found that its eigenvalues are $1/n$ and $a-b=1$, the last one having degeneracy equal to $m-1$ and eigenvectors $\{v_i\}$ in the orthogonal complement of the span of $\textbf{1}_m$. This last property implies that the $\{v_i\}$ are in the kernel of $B$ and, therefore, the $\mathbb{R}^{mn}$ vectors
\begin{equation}
\left\{\left(
              \begin{array}{c}
                v_i \\
                \textbf{0}_m \\
                \vdots \\
                \textbf{0}_m \\
              \end{array}
            \right),\left(
              \begin{array}{c}
                \textbf{0}_m \\
                v_i \\
                \vdots \\
                \textbf{0}_m \\
              \end{array}
            \right),\ldots,\left(
              \begin{array}{c}
                \textbf{0}_m \\
                \vdots \\
                \textbf{0}_m \\
                v_i \\
              \end{array}
            \right)\right\}
\end{equation}
are all eigenvectors of $H^TH$ with eigenvalue equal to 1. Thus, altogether, we have seen that that this eigenvalue has degeneracy at least $1+n(m-1)$. However, $\tr(H^TH)=1+n(m-1)$, which implies that all the remaining eigenvalues necessarily must be equal to 0.
\end{proof}

Now, using Eq.\ (\ref{eqlemma3}) we can write $\langle P,G(m,n)\rangle$ in terms of $\langle P,H\rangle$ and $\langle P,\textbf{1}(m^n,mn)\rangle$. Bounding the former with Theorem 2 and using that the latter equals $nm^n$ for every behaviour $\P$, we immediately arrive at the following result.
\begin{theorem}
Every $(m,n)$ QRAC with $\P\in \q_d$ fulfills
\begin{equation}
p\leq\frac{1}{m}+\frac{\sqrt{md}-1}{m\sqrt{n}}.
\end{equation}
\end{theorem}

This bound reduces to that of \cite{ozols} in the case $m=d=2$. This happens to be tight when $n=2,3$ and the corresponding $\q_2$ behaviours (which cannot be in $\c_2$) have the property that $||P||_1$ fulfills Theorem 1 with equality. One should notice, however, that the bounds are in general not sharp, which can be seen in the particular instance $n=2$ and $d=m>2$ since the exact value of $p$ in this case is known \cite{tavakoli,qracmub,qracmub2}. This can also be seen by comparing with the numerical techniques of \cite{num}.

\subsection{Sharpening witnesses based on state discrimination}

Reference \cite{didwpm} has considered the following scenario: A receives $x\in\x=\{1,\ldots,N\}$ and B $(y,z)\in\x\times\x$ ($y<z$) with the promise that either $x=y$ or $x=z$. The goal of B is to identify which of the two possibilities actually occurred with his output $b\in\{-1,1\}$ (thus $|\x|=N$, $|\y|=N(N-1)/2$ and $|\b|=2$). The aforementioned reference has provided optimal bounds for the performance in this game both in $\c_d$ and $\q_d$ through non-linear ($W_N=\sum_{y<z}(P(1|x=y)-P(1|x=z))^2$) and linear ($V_N=\sum_{y<z}P(1|x=y)-P(1|x=z)$) witnesses. Interestingly, for a fixed value of $d$ there can exist gaps between the classical and quantum bounds; however, this is not the case if $N$ is a multiple of $d$. The optimal quantum strategy consists in sending the states
\begin{equation}\label{states}
|\psi_x\rangle=\frac{1}{\sqrt{d}}\sum_{k=0}^{d-1}\exp\left(i\frac{2\pi kx}{N}\right)|k\rangle
\end{equation}
and B implementing the Helstrom measurement that optimally discriminates between $|\psi_y\rangle$ and $|\psi_z\rangle$ \cite{didwpm,helmstrom}. In the following we show that for $d=2$, the corresponding behaviour is such that $||P||_1=\sqrt{2|\x||\y|}$ $\forall N$. This not only provides another example where Theorem 1 is tight on quantum but non-classical behaviours but, more interestingly, one can then use this insight to construct an alternative witness to $W_N$ and $V_N$ that allows to amplify the gap between the classical and quantum bounds. The general idea on how to construct the witness is very simple. If $P$ has singular value decomposition (SVD) $P=U\Sigma V^T$, then $\langle P,G\rangle=||P||_1$ if $G=UV^T$. Thus, if $\P\in\q_d$ is such that $||P||_1=\sqrt{d|\x||\y|}$, then, by Theorem 2, the witness $G$ constructed following the SVD prescription is optimal for this behaviour. We illustrate this for the behaviour discussed above in the extreme case of $N$ even where the witnesses $W_N$ and $V_N$ cannot discriminate between $\c_2$ and $\q_2$.

After some algebra one finds that in the case $d=2$ the above quantum strategy leads to the behaviour
\begin{equation}\label{optqbeh}
P(b=\pm1|xyz)=\frac{1}{2}\left(1\mp\sin\left[\frac{\pi}{N}(2x-y-z)\right]\right).
\end{equation}
Notice that here we are ignoring the promise that either $x=y$ or $x=z$ and $P\in\mathbb{R}^{N\times N(N-1)}$. In order to compute $||P||_1$, we compute the eigenvalues of $PP^T\in\mathbb{R}^{N\times N}$. Using the above equation we find that the entries of this matrix are given by
\begin{align}\label{calc}
(PP^T)_{xx'}&=\frac{1}{2}\sum_{y<z}\left[1+\frac{1}{2}(\cos\theta_{x-x'}-\cos\theta_{x+x'-y-z})\right]\nonumber\\
&=\frac{N(N-1)}{4}\left(1+\frac{\cos\theta_{x-x'}}{2}\right),
\end{align}
where we have used the shorthand $\theta_j=2\pi j/N$. Since $(PP^T)_{xx'}$ only depends on $|x-x'|$, the matrix is circulant and its eigenvalues $j=0,\ldots,N-1$ are given by
\begin{align}
\lambda_j&=\sum_{k=0}^{N-1}(PP^T)_{1k}\exp\left(i\frac{2\pi jk}{N}\right)=\frac{N(N-1)}{4}\nonumber\\
&\times\left[\sum_{k=0}^{N-1}\exp\left(i\frac{2\pi jk}{N}\right)%\right.\nonumber\\
+\frac{1}{2}\sum_{k=0}^{N-1}\cos\theta_k\exp\left(i\frac{2\pi jk}{N}\right)\right]\nonumber\\
&=\frac{N(N-1)}{4}\left[N\delta_{j0}+\frac{N}{4}(\delta_{j1}+\delta_{j,N-1})\right].
\end{align}
Thus, $P$ has rank 3 $\forall N$ and its non-zero singular values are $N\sqrt{N-1}/2$, $N\sqrt{N-1}/4$ and $N\sqrt{N-1}/4$ amounting to $||P||_1=N\sqrt{N-1}=\sqrt{2|\x||\y|}$.

\begin{theorem}
In the above scenario the witness
\begin{equation}\label{witness}
G(b=\pm1|xyz)=\frac{2}{N\sqrt{N-1}}\left(\frac{1}{2}\mp\sin\left[\frac{\pi}{N}(2x-y-z)\right]\right)
\end{equation}
is such that $\forall\P\in\q_2$
\begin{equation}
\langle P,G\rangle\leq B_Q=N\sqrt{N-1}
\end{equation}
with equality attained by the behaviour given in Eq.\ (\ref{optqbeh}). On the other hand, $\forall\P\in\c_2$ and for even $N$ it holds that
\begin{align}
&\langle P,G\rangle\leq B_C=\frac{2}{N\sqrt{N-1}}\nonumber\\
&\times\left(\frac{N^2(N-1)}{4}+\frac{2}{\sin(\pi/N)}\sum_{y<z}\left|\cos\left[\frac{\pi}{N}(1+y+z)\right]\right|\right),\label{bc}
\end{align}
with equality attained by a deterministic behaviour.
\end{theorem}

Figure 1 plots the ratio $B_C/B_Q$ for different values of $N$. It can be readily seen therein that this is always smaller than 1, quickly approaching its asymptotic value $1/2+4/\pi^2\simeq0.9053$.

\begin{figure}[h]
\includegraphics[scale=0.4]{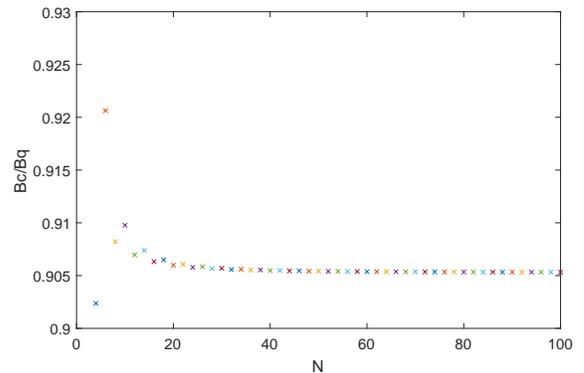}
\caption{Ratio $B_C/B_Q$ for the witness given in Eq.\ (\ref{witness}).}
\end{figure}

\begin{proof}
Proceeding as in Eq.\ (\ref{calc}) on, one finds that $GG^T$ is also a rank 3 matrix with all non-zero eigenvalues equal to 1. Thus, $||G||_\infty=1$ and Theorem 2 gives that $\langle P,G\rangle\leq B_Q$ $\forall\P\in\q_2$. An analogous calculation shows that $\langle P,G\rangle=\tr(PG^T)=B_Q$ for the $\q_2$ behaviour given by Eq.\ (\ref{optqbeh}). Alternatively, the quantum part of the theorem can also be proven by noticing that $G=UV^T$ if the behaviour (\ref{optqbeh}) has the reduced SVD $P=U\Sigma V^T$, i.e.\ $\Sigma=diag(N\sqrt{N-1}/2, N\sqrt{N-1}/4, N\sqrt{N-1}/4)$ and $U$ and $V$ respectively the corresponding $N\times3$ and $N(N-1)\times3$ partial isometries.

It remains to obtain the classical bound $B_C$ for behaviours in $\c_2$. Due to linearity, it must be attained by a deterministic behaviour, i.e.\ such that for every $x,y,z$, $P(b|x,y,z)$ equals 0 or 1 depending on whether $b=\pm1$. Obviously, the best possible strategy is to assign $P(1|xyz)=1$ if $\sin\left[\frac{\pi}{N}(2x-y-z)\right]<0$ and $P(1|xyz)=0$ otherwise. However, Bob does not know $x$ but $m(x)$, which can only take two values (say 0 and 1) given that $\P\in\c_2$. Thus, the best Bob can do is to compute $\sum_{x:m(x)=0}\sin\left[\frac{\pi}{N}(2x-y-z)\right]$ and check the sign of this expression for his inputs $(y,z)$. Since $\sum_{x=1}^N\sin\left[\frac{\pi}{N}(2x-y-z)\right]=0$, the optimal value for a given coding function $m$ is then given by
\begin{align}
&\langle P,G\rangle=\frac{2}{N\sqrt{N-1}}\nonumber\\
&\times\left(\sum_{x,y<z}\frac{1}{2}+2\sum_{y<z}\left|\sum_{x:m(x)=0}\sin\left[\frac{\pi}{N}(2x-y-z)\right]\right|\right).
\end{align}
Noticing now that the best coding function A and B can agree on is that for which most $x$ with the same image lead to $\sin\left[\frac{\pi}{N}(2x-y-z)\right]$ having the same sign for most pairs $(y,z)$, it follows that the optimal strategy corresponds to assigning the same value under $m$ to a consecutive set of elements in $\x$. Thus,
\begin{equation}
B_C=\frac{2}{N\sqrt{N-1}}\left(\frac{N^2(N-1)}{4}+2\max_{j,k}\sum_{y<z}|S_{j,k}|\right),
\end{equation}
where
\begin{align}
S_{j,k}&=\sum_{x=j}^{j+k}\sin\left[\frac{\pi}{N}(2x-y-z)\right]\nonumber\\
&=\frac{\sin\left[\frac{\pi}{N}(2j+k-y-z)\right]\sin\left[\frac{\pi}{N}(k+1)\right]}{\sin(\pi/N)}.\label{maximization}
\end{align}
It is not difficult to show that $\forall j\in\mathbb{Z}$ it holds that
\begin{align}
&\sum_{y=1}^{N-1}\sum_{z=y+1}^N\left|\sin\left[\frac{\pi}{N}(2j+k-y-z)\right]\right|\nonumber\\
&=\sum_{y=1}^{N-1}\sum_{z=y+1}^N\left|\sin\left[\frac{\pi}{N}(k-y-z)\right]\right|,\label{eqsin}\\
&\sum_{y=1}^{N-1}\sum_{z=y+1}^N\left|\cos\left[\frac{\pi}{N}(2j+k-y-z)\right]\right|\nonumber\\
&=\sum_{y=1}^{N-1}\sum_{z=y+1}^N\left|\cos\left[\frac{\pi}{N}(k-y-z)\right]\right|,\label{eqcos}
\end{align}
which we will use repeatedly in the following. Equation (\ref{eqsin}) implies that the value of $j$ is irrelevant in our maximization and, thus, we can write
\begin{equation}
\max_{j,k}\sum_{y<z}|S_{j,k}|=\max_{k}\sum_{y<z}|S_{0,k}|.
\end{equation}
Furthermore, using again Eq.\ (\ref{eqsin}) and Eq.\ (\ref{maximization}) (and taking into account that we are considering $N$ to be even) we can conclude that the above maximum must occur at either $k=N/2-1$ or $k=N/2$. However,
\begin{align}
&\sin(\pi/N)\sum_{y<z}|S_{0,N/2}|\nonumber\\
&=\frac{1}{2}\sum_{y<z}\left|\cos\left[\frac{\pi}{N}(1+y+z)\right]+\cos\left[\frac{\pi}{N}(1-y-z)\right]\right|\nonumber\\
&\leq\frac{1}{2}\sum_{y<z}\left(\left|\cos\left[\frac{\pi}{N}(1+y+z)\right]\right|+\left|\cos\left[\frac{\pi}{N}(1-y-z)\right]\right|\right)\nonumber\\
&=\sum_{y<z}\left|\cos\left[\frac{\pi}{N}(1+y+z)\right]\right|=\sin(\pi/N)\sum_{y<z}|S_{0,N/2-1}|,
\end{align}
where to arrive at the last line we have used Eq.\ (\ref{eqcos}). Hence, the maximum occurs when $k=N/2-1$, i.e.\ the optimal coding function A and B can agree on is one that assigns the same value to a subset of $N/2$ consecutive elements in $\x$. This proves Eq.\ (\ref{bc}).
\end{proof}

As discussed before, the witness $G$ of Theorem 5 can be changed to a witness of the form $G+\sum_{xyz}\alpha_{xyz}A_{xyz}$ keeping track of the corresponding classical and quantum bounds. It might be interesting to notice that this allows to reinterpret its value as the probability of computing some distributed function $f(x,y,z)$. One then has that with probability
\begin{equation}
\pi(x,y,z)=\frac{|\sin\left[\frac{\pi}{N}(2x-y-z)\right]|}{\sum_{x,y<z}|\sin\left[\frac{\pi}{N}(2x-y-z)\right]|}
\end{equation}
A and B receive the inputs $x$ and $(y,z)$ respectively with the goal that Bob answers $f(x,y,z)=1$ ($f(x,y,z)=-1$) whenever $\sin\left[\frac{\pi}{N}(2x-y-z)\right]<0$ ($\sin\left[\frac{\pi}{N}(2x-y-z)\right]>0$). The bounds $B_C$ and $B_Q$ can be changed correspondingly to obtain the maximal average success probability in each setting, which will then be larger in the quantum case.

Here, we have illustrated how the witnesses of \cite{didwpm} can be improved to distinguish classical and quantum bidimensional behaviours. However, different dimensions can be assessed. First, one can directly apply Lemma 2 to conclude that for the witness of Eq.\ (\ref{witness}) it holds
\begin{equation}
\langle P,G\rangle\leq N\sqrt{\frac{d(N-1)}{2}}
\end{equation}
for every $\P\in\q_d$. Another strategy is to consider a different witness taking the one corresponding to the (reduced) SVD of the $\q_d$ behaviour given by Helstrom measurements on the states (\ref{states}) as we did in the $d=2$ case, for which the same bound as above holds.

\section{Conclusion}

Previous works obtaining bounds for DIDW in the prepare-and-measure scenario had relied on particular constructions of dimension witnesses or considered arbitrary behaviours under the condition that shared randomness among parties is not available. In this work we have obtained general and explicit lower bounds on the dimension of arbitrary quantum behaviours dropping this assumption, which are based on standard techniques from matrix theory. Although the bounds are not always tight and in general perform worse than the numerical techniques of Ref.\ \cite{num}, we expect that their simple and easy-to-use form make them helpful for further investigations in this context, particularly when a numerical approach is not feasible due to the analytical nature of the problem at hand or because it is computationally too demanding. In fact, we have provided two applications of our result. First, we have proved that using our techniques the probability of success of distributed computational tasks can be upper bounded in general as a function of the dimension of the message in a paradigmatic example such as random access codes. Second, we have shown that our construction allows one to derive powerful dimension witnesses for given behaviours. In particular, when the behaviour $\P$ is at the boundary of $\q_d$ and fulfills the condition of Theorem 1 with equality, the SVD of the matrix associated to $\P$ yields immediately an optimal dimension witness for it. As an example, we have used this to improve the constructions of dimension witnesses given in \cite{didwpm}. For the future we hope that this insight makes it possible to find systematically adequate dimension witnesses in relevant physical situations, to improve known semi-device-independent protocols and, in general, to understand better the mathematical structure of fixed-dimensional behaviours. It might be also interesting to apply these techniques to bound the success probability as a function of the allowed quantum communication for other distributed tasks of interest.

%\begin{acknowledgments}
This research was funded by the Spanish MINECO through grants MTM2017-84098-P and MTM2017-88385-P and by the Comunidad de Madrid through grant QUITEMAD+CM S2013/ICE-2801.
%\end{acknowledgments}

\end{document}